\documentclass{article}

\usepackage{PRIMEarxiv}

\usepackage[utf8]{inputenc} 
\usepackage[T1]{fontenc}    
\usepackage{hyperref}       
\usepackage{url}            
\usepackage{booktabs}       
\usepackage{amsfonts}       
\usepackage{nicefrac}       
\usepackage{microtype}      
\usepackage{lipsum}
\usepackage{fancyhdr}       
\usepackage{graphicx}       
\graphicspath{{media/}}     

\usepackage{caption}
\usepackage{subcaption}
\usepackage[ruled,linesnumbered]{algorithm2e}
\usepackage{breqn}
\usepackage{tabularx}
\usepackage{amsthm}
\newtheorem{theorem}{Theorem}[section]
\newtheorem{lemma}[theorem]{Lemma}

\title{Some results on Minimum Consistent Subsets of Trees

}

\author{
  Bubai Manna \\
  Department of Mathematics \\
  Indian Institute of Technology Kharagpur \\
  Kharagpur, 721302\\
  \texttt{bubaimanna11@gmail.com} \\
   \And
  Bodhayan Roy \\
  Department of Mathematics \\
  Indian Institute of Technology Kharagpur \\
  Kharagpur, 721302\\
  \texttt{bodhayan.roy@gmail.com} \\
}

\begin{document}

\maketitle

\begin{abstract}
For a graph \emph{G = (V,E)} where each vertex is coloured by one of the colours $\{c_1, c_2, . . . , c_k\}$, consider a subset $C \subseteq V$ such that for each vertex $v \in V\backslash C$, its set of nearest neighbours in \emph{C}  contains at least one vertex of the same colour as that of \emph{v}. Such a \emph{C} is called a \emph{consistent subset (CS)}. Computing a consistent subset of the minimum size is called the \emph{Minimum Consistent Subset (MCS)} problem. \emph{MCS} is known to be NP-complete for planar graphs. We propose a polynomial-time algorithm for finding a minimum consistent subset of a \emph{k-chromatic} spider graph when \emph{k} is a constant and $k\ge 3$. We also show \emph{MCS} remains \emph{NP-complete} on trees.
\end{abstract}


\section{Introduction}

Let $P=\{p_1, p_2, \dots, p_n\}$ be a set of \emph{n} points on a plane, and $\{c_1, c_2, . . . , c_k\}$ be a set \emph{k} colours. Each point $p_i$ is coloured by one of the colours of the set $\{c_1, c_2, . . . , c_k\}$. A subset $C \subseteq P$ is called a consistent subset of $P$ if and only if for each point $v \in V\backslash C$, its set of nearest neighbours in $C$ contains at least one point of the same colour as that of $v$. 

Let \emph{G=(V, E)} be a graph where \emph{V} and \emph{E} be the vertex set and edge set, respectively. The distance between a pair of vertices $u$ and $v$ of a graph \emph{G=(V, E)} is the number of edges in the shortest path from $u$ to $v$, and the length of the shortest path will be referred to as \emph{hop-distance(u, v)}. For any set $C\subseteq V$ and a vertex $v\in V\backslash C$, the \emph{nearest neighbours} of \emph{v} in \emph{C} is the nearest vertices in \emph{C} with respect to \emph{hop-distances}. The definition of \emph{MCS} for graphs is as follows.

\begin{center}
\begin{tabular}{ | m{15cm} |} 
  \hline
   \begin{center} Let \emph{G=(V, E)} be a graph and each vertex of \emph{V} is coloured by one of the \emph{k} colours of the colour set $\{c_1, c_2, \dots, c_k\}$. We classified \emph{V} into $k$ classes, namely $V1$, $V_2$, . . . , $V_k$ where all the vertices of each class $V_i$ have a colour from the set $\{c_1, c_2, \dots, c_k\}$, and  no two $V_i$ and $V_j$ ($i\ne j$) have same colour vertices. The objective is to choose subsets $V_i^{'} \subseteq V_i$, $i = 1, . . . , k$ such that for each member $v \in V$, if $v \in V_i$ then among its nearest neighbours in $\cup_{i=1}^{k}V_i^{'}$ there is a vertex $u$ of $V_i^{'}$, and $\sum_{i=1}^{k} \lvert V_i^{'} \rvert$ is minimum. In other words, for a graph \emph{G=(V, E)}, a subset $C\subseteq V$ is called a \emph{minimum consistent subset} if, for any vertex $v\in V\backslash C$, the set of nearest neighbours of $v$ in $C$ has a vertex $u$ of the same colour as that of $v$ and $\lvert C \rvert$ is minimum. Such a vertex $u$ is called the \emph{covering vertex} of $v$. In other words, we say that $v$ is covered by $u$.\end{center}  \\ 
   \hline 

   \end{tabular}
\end{center}

\subsection{Background}
Hart \cite{Hart} first introduced the main idea for the geometric variation of a consistent subset problem. Later, it is also introduced to find a minimum size consistent subset for a point set on a plane. Wilfong \cite{Wilfong} and \cite{Wilfong1} showed that the MCS problem is NP-complete even for three-coloured point sets in $\mathbb{R}^{2}$. The author also proposed an $O(n^2)$ time algorithm for the minimum consistent subset problem with two coloured point sets where one set is a singleton. In \cite{Bodhayan}, it is also proved that the minimum consistent subset problem is NP-hard for bicoloured points in $\mathbb{R}^{2}$.  MCS for collinear point sets in $\mathbb{R}^{2}$ can be found in $O(n^2)$ time \cite{Banerjee}. Recently, Biniaz (\cite{Biniaz}, \cite{Biniaz1}) introduced a sub-exponential time algorithm for the set of points in $\mathbb{R}^{2}$ for the consistent subset problem. In the papers, \cite{Sanjana} and \cite{Anil}, the authors proposed a polynomial-time algorithm for bicoloured paths, cycles, spiders, caterpillars, combs, and trees. Their algorithm for bicoloured paths also works for k-chromatic paths.

\subsection{Our Contribution} MCS is known to be NP-complete for general and even planar graphs. This paper presents (i) a polynomial-time algorithm for the MCS problem for k-chromatic spider graphs when $k$ is a constant and $k\ge 3$. This algorithm runs in $O(n^{k+2})$ time, where $n$ is the number of vertices in the spider. (ii) We also show that MCS remains NP-complete on trees. In the rest of the paper, we will use $C$ to denote a minimum consistent subset of the input graph $G$.


\section{MCS on k-chromatic spider graphs}

Let \emph{G=(V,E)} be a spider graph, where \emph{v} is the center of the spider graph and $\lvert V \rvert=n$. Let $k_1$ be the total number of legs of the spider and the legs are denoted by $L_1, L_2,\dots, L_{k_1}$. $V_i=\{v_{1,i}, v_{2,i}, v_{3,i}, \dots , v_{\lvert V_i \rvert,i}\}$ and $E_i=\{v_{j,i}v_{j+1,i} : j= 1, 2, 3, \dots, \lvert V_i \rvert-1\}$ are the vertex set and edge set of the \emph{i}-th leg of the spider, respectively. The vertex $v_{1,i}$ of the \emph{i}-th leg is connected with the center \emph{v} by an edge $v_{1,i}v$ for $i=1,2,3, \dots, k_1$. Hence the vertex set and the edge set of \emph{G} are $V=\{v\}\bigcup \cup_{i=1}^{k_1}V_i$ and $E=\bigcup_{i=1}^{k_1}(E_i\cup \{v_{1,i}v\})$, respectively. Also, the cardinality of the vertex set and edge set of \emph{G} are $\lvert V \rvert = \sum_{i=1}^{k_1} \lvert V_i \rvert + 1= n$ and $\lvert E \rvert= \cup_{i=1}^{k_1}\lvert E_i \rvert+k_1= \sum_{i=1}^{k_1} \lvert V_i \rvert = n-1$, respectively. Each vertex of \emph{V} is coloured by one of the colours of the set $\{c_1, c_2,\dots, c_k\}$. Hence, the \emph{k-chromatic} spider graph is defined.

A consecutive set of vertices of the same colour on a leg is called a \emph{run}. For $L_i$ (\emph{i}-th leg), the $run$ connected to $v$ is called the \emph{first run}, and it is denoted by $\rho_i$, and the subsequent runs are $R_{2,i}, R_{3,i}, \dots$. We denote $C(u)$ as a minimum consistent subset of $G$ among all possible consistent subsets that contain a vertex \emph{u}. If $p$ is a vertex of the leg $L_i$, then we write $leg(p)=L_i$. For a \emph{first run} $\rho_i$ of the leg $L_i$, we write \emph{$L_i$ = leg($\rho_i$)}. $V(\rho_i)$ is the set of vertices of a leg whose \emph{first run} is $\rho_i$. The colour of the vertices of any run $R_{j,i}$ is denoted by \emph{col($R_{j,i}$)}, and it satisfies \emph{col($R_{j,i}$) $\in \{c_1, c_2, c_3, \dots, c_k\}$}. Also, the colour of a vertex $p$ is denoted by \emph{col(p)}, and it satisfies \emph{col(p) $\in \{c_1, c_2, c_3, \dots, c_k\}$}. We use $C\subseteq V$ as a \emph{minimum consistent subset} of the spider graph \emph{G}.

\begin{lemma} \label{path}
    The run time to find out a minimum consistent subset of the k-chromatic path of n vertices is $O(n)$.
\end{lemma}
\begin{proof}
    The proof of the lemma is in the paper \cite{Sanjana}.
\end{proof}

\begin{lemma}\label{bubai1}
If the center of the spider is not in $C$, then the vertex covering of the center must be in the union of the center and the \emph{first runs}.

\end{lemma}

\begin{figure}
\centering
\includegraphics[width=.9\textwidth]{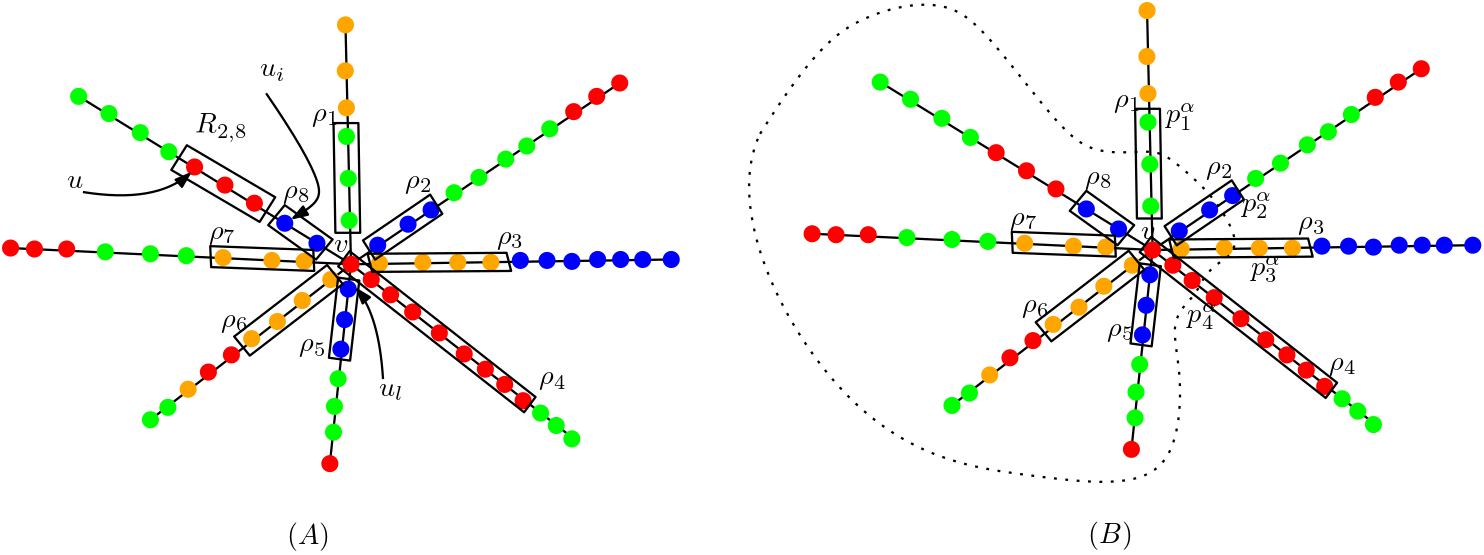}
\caption{\label{fig:case3A}(A) In this figure $k=8$, $j=2$ and $l=5$. (B) Here $k_2=k=4$ and $\psi = \{c_1 =$ $'green'$ $, c_2 = $ $'blue',$ $ c_3 = $ $'orange', $ $ c_4 = $ $ 'red'\}$ is the complete list of $k$ colours. $S$=$\{\rho_1, \rho_2, \rho_3, \rho_4\}$ such that, $\cup_{\rho_i \in S}col(\rho_i)$= $\psi$. $\alpha \in \{1,2, \dots, $min\{length of each \emph{first run} of $S \}$=\{1,2,3\}.}
\end{figure}

\begin{proof}

The \emph{minimum consistent subset} \emph{C} does not contain the center of the spider \emph{v}, so \emph{v} must be covered by a vertex of \emph{C}. Let $u\in C$ be the vertex of the \emph{j}-th run $R_{j,k}$ ($j\ge 2$) of the leg $L_k$ such that \emph{u} covers \emph{v} (Fig.\ref{fig:case3A}(A)). As \emph{u} covers \emph{v} and \emph{u} is a vertex of the \emph{run} $R_{j,k}$, then \emph{col(v) = col(u) = col($R_{j,k}$)}. Let $U$ be the path from \emph{v} to \emph{u}, excluding \emph{v} and \emph{u}. As \emph{col(v) = col(u) = col($R_{j,k}$)} and $R_{j,k}$ is not the \emph{first run}, then it is obvious that there must have been at least one \emph{run} $R_{i,k}$ on the path \emph{U} such that \emph{col($R_{j,k})\ne$ col($R_{i,k}$)}. Let $u_i$ be a vertex of the \emph{run} $R_{i,k}$ as well a vertex of \emph{U}. As \emph{v} is covered by \emph{u}; then it implies that \emph{u} is the nearest vertex of \emph{v} in \emph{C}. As \emph{u} is the nearest vertex of \emph{v} in \emph{C}, then \emph{u} is the nearest vertex of all the vertices of \emph{U} in \emph{C}. Hence, all the vertices of \emph{U} must be covered by \emph{u}. As $u_i$ is a vertex of \emph{U}, then $u_i$ is also covered by \emph{u}, but $col(u_i)\ne col(u)$, a contraction.

\end{proof}

Depending on the situation, we have a total of three situations:

(1) The colour of $v$ differs from the colour of each $\rho_i$ for $i=1, 2,\dots, k_1 $.

(2) $col(\rho_i)$ is of the same colour as that of $v$ for $i = 1, 2, \dots , k_1$ . 

(3) $col(v) \in \cup_{i=1}^{k_1} col(\rho_i)$ such that $2\leq \lvert \cup_{i=1}^{k_1} col(\rho_i) \rvert \leq k$.

\subsection{Case 1}\label{case1}

 We have \emph{col(v) $\ne col(\rho_i)$ } for $i=1, 2,\dots, k_1$ (Fig.\ref{fig:case1}(A)). The center \emph{v} must be included in \emph{C}; otherwise, the vertex covering of \emph{v} must be a vertex of a \emph{first run} $\rho_i$ using lemma \ref{bubai1}, but $col(\rho_i)\ne col(v)$ for $i=1,2,\dots, k_1$. So,  We compute the minimum consistent subset $C_i$ using the path graph algorithm using the lemma \ref{path} for each path $L_i$ $\cup \{v\}$ by assuming $v\in C_i$. Later, we remove $v$ from $C_i$ and compute the minimum consistent $C=$ $ \cup_{i=1}^{k_1} $ $ C_i \bigcup \{v\}$. We set $\chi = \sum_{i=1}^{k_1} \lvert C_i \rvert +1$. Lemma \ref{path} takes $O(\lvert V_i \rvert+1)$ time to execute each $C_i$. Hence, the total execution time is $O(n)$.

\begin{figure}
\centering
\includegraphics[width=.9\textwidth]{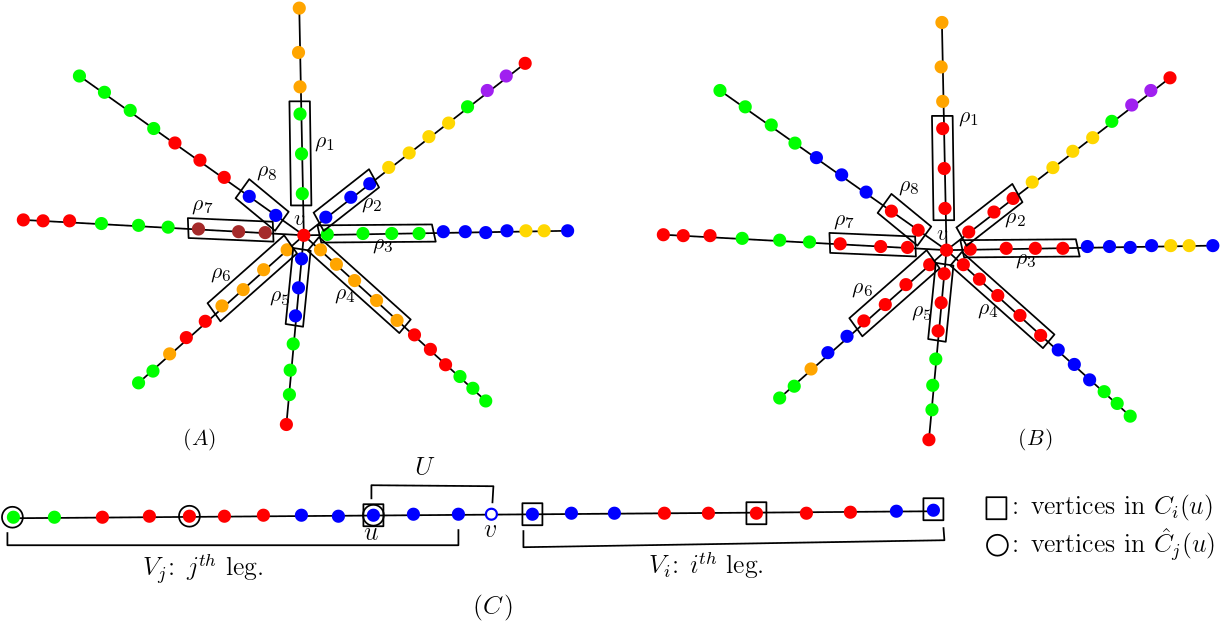}
\caption{\label{fig:case1} (A) $col(v)$ is different from each $\rho_i$ for $i=1,2,\dots,8$. Here it is obvious that $v\in C$; otherwise, the consistency property for $v$ is not satisfied. (B) $col(v)$ =$col(\rho_i)$ for $i=1,2, \dots, 8$. Here the vertex $v$ may or may be in $C$. (C) $C_i(u)$: optimum solution for $L_i \cup U$, and $\hat C_{j}(u)$: optimum solution for $L_j\backslash U \cup \{u\}$.}
\end{figure}

\subsection{Case 2}\label{case2}
 Without loss of generality, let \emph{col(v)=col($\rho_i$)= $c_j$} for all $i=1,2,\dots, k_1$ (Fig.\ref{fig:case1}(B)). We initialize $C$ and $\chi$ using the previous case \ref{case1}; i.e.,$\chi = \sum_{i=1}^{k_1} \lvert C_i \rvert +1$. Next, we test whether the size of the minimum consistent subset can be improved if \emph{v} is not chosen in \emph{C}.

 Let $u \in \cup_{j=1}^{k_1}\rho_j$ such that $u\in \rho_j$, and $U$ be the path segment from $v$ to $u$. We use (i) $C_i(u)$ to denote the minimum sized consistent subset of $L_i \cup U$ where \emph{u} is included in \emph{$C_i(u)$}, and it is the vertex in $C_i(u)$ closest to \emph{v}, and (ii) $\hat C_j(u)$ is the minimum sized consistent subset of the path $(L_j \backslash U)\cup \{u\}$ that includes $u$ (Fig.\ref{fig:case1}(C)).

\begin{lemma}
    If $v\notin C$, then there exists at least one vertex $u \in \cup_{i=1}^{k_{1}} \rho_i$ in $C$. If the vertex $u\in \rho_j$ belongs to $C$ and is closest to $v$ with respect to hop-distance, then $C= \hat C_{j}\bigcup (\cup_{i=1,2,\dots,k_1, i\ne j}$ $C_i(u))$.
\end{lemma}
\begin{proof}

We first compute $C_i(u)$ the lemma \ref{path} for the path $L_i\cup U$. $C_i(u)$ may contain another vertex $w\in \rho_i$ (in addition to $u$). In such a case, as $u$ is closest to $v$ among the vertices in $C_i(u)$, $w$ must satisfy hop-distance$(v, u) \le$ hop distance$(w, u)$. While constructing the \emph{overlay graph} (note that the \emph{overlay graph} is defined in the paper \cite{Sanjana}), we add the edges $(w, u)$ for $w\in \rho_i$ satisfying hop-distance$(v, u) \le$ hop-distance$(w, u)$. If there exists any edge $(w, u)$ with $col(w)= col(u)$ (from the run of $L_i$ adjacent to $\rho_i$ to the vertex $u$), we must have hop-distance$(v, u) \le$ hop-distance$(w, u)$ (for the consistency of vertex $v$ in $C_i(u)$). Needless to say, no edge $(w, u)$ will be present in the graph for $w\in U$. Now, the shortest path in this overlay graph will produce a minimum size consistent subset $C_i(u)$ for $L_i\cup U$. Remaining $\hat C_j(u)$ gives the minimum consistent subset for $L_i\backslash U$.
\end{proof}

The run time to find out such $u$ nearest to $v$ in the \emph{first runs} such that $u\in C$ is $O(n)$, and the execution time to calculate each $C_i(u)$ takes at most $O(n)$ for each \emph{i}. Also, the execution time to calculate $\hat C_j(u)$ is $O(n)$. Hence the total time to execute the whole process takes $O(n^2)$ time.

\subsection{Case 3}\label{case3}

Without loss of generality, we assume that $\cup_{i=1}^{k_1} col(\rho_i)=\{c_1, c_2, \dots, c_{k_2}\}$, where $2\leq k_2 \leq k$, and $col(v)={c_1}$. If $k_2=1$, then the center $v$ and all the \emph{first runs} $\rho_i$ for $i=1,2,\dots, k_2$ have the same colour, which is already discussed in subsection \ref{case2}. Depending on the value of $k_2$ we have two cases: 

(1) $k_2=k$, which means vertices of all $k$ colours are available in $\cup_{i=1}^{k_1} col(\rho_i)$.

(2) $2 \leq k_2 < k$, which means that not all the $k$ colours are available in $\cup_{i=1}^{k_1} col(\rho_i)$.

\textbf{Subcase 1:} It is obvious that $k_1 \geq k$ because there are a total of $k$ colours available in $k_1$ \emph{first runs} (Fig.\ref{fig:case3A}(B)). Without loss of generality, let the first $k$ runs be of different colours, that is $col(\rho_i)$=$c_i$ for $i=1, 2, \dots, k$, and $\cup_{i=1}^{k} col(\rho_i)=\{c_1, c_2, \dots, c_{k}\}$. We take $k$ vertices $\{p_i : p_i\in \rho_i$ for $i=1,2,\dots, k\}$ in $C$ in such a way that \emph{hop-distance$(p_1, v)$ = hop-distance$(p_2, v)$= \dots = hop-distance$(p_{k}, v)$ = $d_{\alpha}$} (we say the equal distance is $d_{\alpha}$), then all the vertices in $V \backslash \{V_1 \cup V_2 \cup \dots \cup V_{k}\}$ are covered with respect to their consistency. We must add the minimum consistent subsets $\hat C_i(p_i)$ for $L_i\backslash U_i$ (note that, $U_i$ is the path from $p_i$ to $v$) that includes $p_i$ for $i=1,2,\dots,k$. The size of the consistent subset $\chi$ will be updated if the existing $\chi > \sum_{i=1}^{k} \lvert \hat C_i(p_i) \rvert$. To find out $\hat C_i(u)$, we need to use case \ref{case2}, where $u$ is given. We will repeat the process for each such distance $d_{\alpha}$ $ \in \{1, 2, . . . $ $,   min\{\lvert \rho_1 \rvert \, \lvert \rho_2 \rvert \, \dots, \lvert \rho_k \rvert \}\}$. A single execution of the algorithm for the path (lemma \ref{path}) returns the size of all the consistent subsets of $L_i\backslash U_i$ with every vertex $p_{\alpha}$, $\alpha \in \{1, 2, . . . $ $,   min\{\lvert \rho_1 \rvert \, \lvert \rho_2 \rvert \, \dots, \lvert \rho_k \rvert \}$ as its element that is closest to v for all $i=1,2,\dots, k$. We run the algorithm of lemma \ref{path} for every leg $L_i$ ($i=1,2,\dots, k_2$) of the spider, and the total time needed is $O(\lvert V \rvert)$=$O(n)$. While choosing $d_{\alpha} \in \{1,2,\dots, \lvert \rho_i \rvert \}$ we take total time $O(n)$.

Now, $\chi$ may not be a minimum just because we took a particular set of $k$ \emph{first runs}, which satisfies $\cup_{i=1}^{k} col(\rho_i)=\{c_1, c_2, \dots, c_{k}\}$. So, we have to take all the sets of $k$ \emph{first runs} whose colour set is $\{c_1, c_2, \dots, c_{k}\}$. So, we find all possible $k$ \emph{first runs} from $k_1$ \emph{first runs} whose colour set is $\{c_1, c_2, \dots, c_{k}\}$, and then we repeat the above procedure. Finally, we get the minimum $\chi$, and the corresponding set is our required minimum consistent subset, which executes in the time $O($ $\binom{k_1}{k}$$*n*n)$ $ \approx O(n^{k+2})$ (note that the number of vertices $n$ is always greater than or equal to the number of legs).  We explain the whole procedure in algorithm \ref{alg:subcase01}:

\begin{algorithm}
\caption{Algorithm For Subcase 1}\label{alg:subcase01}
\KwData{$|V|=n$, $k_1$ legs $\{L_1, L_2, \dots, L_{k_1}\}$ and $k_1$ \emph{first runs} $\{\rho_1, \rho_2,\dots, \rho_{k_1}\}$, a total of $k$ colours $\psi=\{c_1, c_2,\dots, c_k\}$, $\hat C_i(p_{\alpha})$ is defined in \ref{case2}. $C=$ $ \cup_{i=1}^{k_1} $ $ C_i \bigcup \{v\}$ and $\chi = \sum_{i=1}^{k_1} \lvert C_i \rvert +1$ are taken from case \ref{case1}}
\KwResult{Minimum Consistent Subset $C$}

\For{$S$= $k$ \emph{first runs} from $k_1$ \emph{first runs}}
{
  \If{$\cup_{\rho_i\in S} $ $col(\rho_i)$=$\psi$}
    {
      \For{ each $\alpha \in \{1,2, \dots, $min\{length of each \emph{first run} of $S \}\}$}
      {
        Take vertex $p_{\alpha}$ from each $\rho_i$ of $S$ such that hop-distance($p_{\alpha},v)$ are all equal for every $i$;

        Calculate $\hat C_i(p_{\alpha})$ for all $k$ \emph{first runs} of $S$ from \ref{path};

        $C_{1}$= Union of $\hat C_i(p_{\alpha})$ for all $k$ \emph{first runs} of $S$;
        
        $\chi_{1}$= $|C_1|$;         
        
        \eIf{$\chi > \chi_{1}$ }
         {
         $C$= $C_1$;
         
         $\chi$= $\chi_1$;
         }
          {
           $C$= $C$;
           
           $\chi$=$\chi$;
          }  
        
      }
    }

}

\end{algorithm}

\textbf{Subcase 2:} In this case, not all the $k$ colours are available in $\cup_{i=1}^{k_1} col(\rho_i)$. We assume $\psi=\{c_1, c_2, \dots, c_{k_2}\}$, and $col(v)\in \psi$.  We have $\cup_{i=1}^{k_1} col(\rho_i)=\psi$, where $2\leq k_2 < k$. Depending on the situation, we have two cases:

(A) Without loss of generality, there exists $k_2$ \emph{first runs} $\rho_1$, $\rho_2$, \dots, $\rho_{k_2}$ for which $\cup_{i=1}^{k_2}$ $ col(\rho_i) = \psi$ such that for every vertex $p \in V\backslash $ $(\cup_{i=1}^{k_2}V_{i})$, the colour of the vertex $p$ satisfies $col(p)\in \psi$.

(B) Without loss of generality, there does not exist $k_2$ \emph{first runs} $\rho_1$, $\rho_2$, \dots, $\rho_{k_2}$ for which $\cup_{i=1}^{k_2} col(\rho_{i})$ = $\psi$ such that for every vertex $p \in V\backslash (\cup_{i=1}^{k_2}V_i)$, the colour of the vertex $p$ satisfies $col(p) \in \psi$

\begin{figure}
\centering
\includegraphics[width=.9\textwidth]{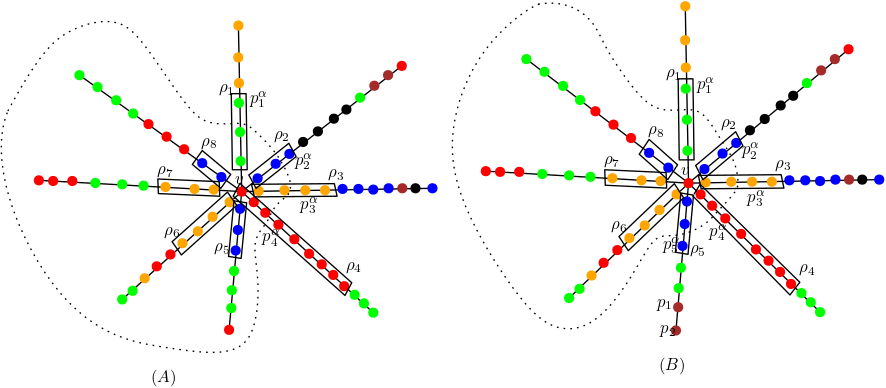}
\caption{\label{fig:case3B1}(A): Here $k=6$, $k_2=4$ and $\{c_1 =$ $'green'$ $, c_2 = $ $'blue',$ $ c_3 = $ $'orange', $ $ c_4 = $ $ 'red' $  $c_5 = $ $'brown'$ $c_6 = $ $'black'\}$ is the complete list of $k$ colours. $\psi = \{c_1 =$ $'green'$ $, c_2 = $ $'blue',$ $ c_3 = $ $'orange', $ $ c_4 = $ $ 'red'\}$ is the complete list of $k_2$ colours. $S$=$\{\rho_1, \rho_2, \rho_3, \rho_4\}$ such that, $\cup_{\rho_i \in S}col(\rho_i)$= $\psi$ and there does not exist any $p \in V\backslash (\cup_{\rho_i \in S} V(\rho_i))$ for which $col(p) \notin \psi$. (B): Here $k=6$, $k_2=4$ and $\{c_1 =$ $'green'$ $, c_2 = $ $'blue',$ $ c_3 = $ $'orange', $ $ c_4 = $ $ 'red' ,$  $c_5 = $ $'brown',$ $c_6 = $ $'black'\}$ is the complete list of $k$ colours. $\psi = \{c_1 =$ $'green'$ $, c_2 = $ $'blue',$ $ c_3 = $ $'orange', $ $ c_4 = $ $ 'red'\}$ is the complete list of $k_2$ colours. $S$=$\{\rho_1, \rho_2, \rho_3, \rho_4\}$ such that, $\cup_{\rho_i \in S}col(\rho_i)$= $\psi$ and there exist at least one $p_1 \in V\backslash (\cup_{\rho_i \in S} V(\rho_i))$ for which $col(p_1) \notin \psi$. $\phi_1 = \{c_2\}$, $\phi_2 = \{c_1, c_3, c_4\}$, $S_1=\{\rho_5\}$, $S_2=\{\rho_1, \rho_3, \rho_4\}$, $S_3 =\{\rho_2\}$.}
\end{figure}

\textbf{Subcase 2A:} We have $\cup_{i=1}^{k_2}$ $ col(\rho_i) = \psi$. We assume col($\rho_i$)=$c_i$ for $i=1, 2, \dots, k_2$. We take the vertices $\{p_i : p_i\in \rho_i$ for $i=1,2,\dots, k_2\}$ in $C$ in such a way that hop-distance$(p_1, v)$ = hop-distance$(p_2, v)$= \dots = hop-distance $(p_{k_2}, v)$ = $d_{\alpha}$, then all the vertices in $V \backslash (V_1 \cup V_2 \cup \dots \cup V_{k_2})$ are covered with respect to their consistency (Fig.\ref{fig:case3B1}(A)). We repeat the above \textbf{subcase 1}, and we update $C$ if we get the smallest size \emph{consistent subset}. We repeat the process by taking every possible $k_2$ \emph{first runs}, which satisfies the conditions of this case. We explain the procedure of this case in the algorithm \ref{alg:subcase02}. Hence, we get a minimum consistent subset $C$ with the time complexity $O(n^{k+2})$.

\begin{algorithm}
\caption{Algorithm For Subcase 2A}\label{alg:subcase02}
\KwData{$|V|=n$, $k_1$ legs $\{L_1, L_2, \dots, L_{k_1}\}$ and $k_1$ \emph{first runs} $\{\rho_1, \rho_2,\dots, \rho_{k_1}\}$, a total of $k$ colours $\psi=\{c_1, c_2,\dots, c_{k_2}\}$ with $2 \leq k_2< k$, $\hat C_i(p_{\alpha})$ is defined in \ref{case2}. $C=$ $ \cup_{i=1}^{k_1} $ $ C_i \bigcup \{v\}$ and $\chi = \sum_{i=1}^{k_1} \lvert C_i \rvert +1$ are taken from case \ref{case1}}
\KwResult{Minimum Consistent Subset $C$}

Condition 1: There does not exist any $p \in V\backslash \cup_{\rho_i \in S} V(\rho_i)$ for which $col(p) \notin \psi$

\For{$S$= $k_2$ \emph{first runs} from $k_1$  \emph{first runs} which satisfies $\cup_{\rho_i\in S} $ $col(\rho_i)$=$\psi$ and the condition 1}
{
  
      \For{ each $\alpha \in \{1,2, \dots, $min\{length of each \emph{first run} of $S \}\}$}
      {
        Take vertex $p_{\alpha}$ from each $\rho_i$ of $S$ such that hop-distance($p_{\alpha},v)$ are all equal for every $i$;

        Calculate $\hat C_i(p_{\alpha})$ for all $i$ from \ref{path};

        $C_{1}$= Union of $\hat C_i(p_{\alpha})$ for all $i$;
        
        $\chi_{1}$= $|C_1|$;         
        
        \eIf{$\chi > \chi_{1}$ }
         {
         $C$= $C_1$;
         
         $\chi$= $\chi_1$;
         }
          {
           $C$= $C$;
           
           $\chi$=$\chi$;
          }  
        
      }

}

\end{algorithm}

\textbf{Subcase 2B:} We have $\cup_{i=1}^{k_1}$ $ col(\rho_i) = \{c_1, c_2,\dots, c_{k_2}\}= \psi$. According to this sub-case, if we take any $k_2$ \emph{first runs} in the set $S= \{\rho_1$, $\rho_2$, \dots, $\rho_{k_2}\}$ from $k_1$ \emph{first runs} $\rho_1$, $\rho_2$, \dots, $\rho_{k_1}$ such that $\cup_{i=1}^{k_2} col(\rho_i)$ = $\{c_1, c_2,\dots, c_{k_2}\}$, then there exists at least one $p \in V\backslash (\cup_{i=1}^{k_2}V_i)$ for which $col(p) \notin \{c_1, c_2, $ $\dots$$, c_{k_2}\}$ (Fig.\ref{fig:case3B1}(B)). Without loss of generality, assume that col($\rho_i$)=$c_i$ for $i=1, 2, \dots, k_2$.

\begin{algorithm}
\caption{Algorithm For Subcase 2B}\label{alg:subcase03}
\KwData{$|V|=n$, $k_1$ legs $\{L_1, L_2, \dots, L_{k_1}\}$ and corresponding vertex set $\{V_1, V_2, \dots, V_{k_1}\}$ and $k_1$ \emph{first runs} $\{\rho_1, \rho_2,\dots, \rho_{k_1}\}$, a total of $k_2$ colours $\psi=\{c_1, c_2,\dots, c_{k_2}\}$ with $2 \leq k_2< k$, $\hat C_i(p_{\alpha})$ is defined in \ref{case2}. $C=$ $ \cup_{i=1}^{k_1} $ $ C_i \bigcup \{v\}$ and $\chi = \sum_{i=1}^{k_1} \lvert C_i \rvert +1$ are taken from case \ref{case1}}
\KwResult{Minimum Consistent Subset $C$}

Condition 2: There exists at least one $p \in V\backslash (\cup_{\rho_i \in S} V(\rho_i))$ for which $col(p) \notin \psi$

\For{$S$= $k_2$ \emph{first runs} from $k_1$ \emph{first runs} which satisfies condition 2}
{
      $\phi$=$\{ \}$;

      \For{ each $p_j\in V\backslash \{\cup_{\rho_i \in S}V(\rho_i)\}$ which satisfies $col(p_j) \notin \psi$}  
      {
       \If{$p_j$ is the vertex of the leg $L_j$}
       {

       $\phi$=$\phi \cup \{L_j\}$ (Which implies that $L_j$ is an element of $\phi$);
       
       }

      }
      $\psi_{1}$= $\cup_{L_j \in \phi} col(\rho_j)$;
      
      $\psi_2$=$\psi \backslash \psi_{1}$;
      
      $\psi^{'}$=$\psi \backslash \psi_{2}$;

      $S_1$=Set of \emph{first runs} of the leg $L_j$ if $V_j \in \phi$;
      
      $S_2 = S\backslash S_{1}$;

      $S^{'}=S\backslash S_2$;

      $S_3= \{ \}$;
      
      \For{each leg $L_j$ corresponding to each run $\rho_j$ of $S^{'}$}
      {

      \If{ there exists a vertex $q_j\in V_j$ such that col$(q_j) \notin \phi$}
      {

      $S_3=S_3 \cup \{\rho_j\}$;

      }

      }

      $S_0=S_1 \cup S_2 \cup S_3$;
      
      \For{ each $\alpha \in \{1,2,\dots, $min\{length of each \emph{first run} of $S_0$  $\}\}$}
      {
        Take vertex $p_{\alpha}$ from each $\rho_i$ of $S_0$ such that hop-distance($p_{\alpha},v)$ are all equal for every $i$;

        Calculate $\hat C_i(p_{\alpha})$ for all $i$ from \ref{path};

        $C_{1}$= Union of $\hat C_i(p_{\alpha})$ for all $i$;
        
        $\chi_{1}$= $|C_1|$;         
        
        \eIf{$\chi > \chi_{1}$ }
         {
         $C$= $C_1$;
         
         $\chi$= $\chi_1$;
         }
          {
           $C$= $C$;
           
           $\chi$=$\chi$;
          }  
        
      }

}

\end{algorithm}

Without loss of generality, let there be a total of $k_3$ vertices $q_{k_2+1}, q_{k_2+2},\dots ,q_{k_2+k_3}$ in $V\backslash (\cup_{i=1}^{k_2}V_i)$ such that $col(q_{k_2+i}) \notin \{c_1, c_2, \dots, c_{k_2}\}$ for $i=1,2,\dots, k_3$. These vertices $q_{k_2+1}, $ $ q_{k_2+2},\dots$$,q_{k_2+k_3}$ must be the vertices of some legs except the legs $L_1$, $L_2$,\dots, $L_{k_2}$. Without loss of generality, we assume that $\phi=\cup_{i=1}^{k_3} $leg$(q_{k_{2}+i})$=$\{L_{k_2+1}, $ $ L_{k_2+2}, $ $ \dots, $ $ L_{k_2+k_4}\}$ and corresponding \emph{first runs} of $\phi$ is in the set $S_1= \{\rho_{k_2+1}$, $\rho_{k_2+2}$, \dots, $\rho_{k_2+k_4}\}$. Let $\psi_1$ be the set of colours of the \emph{first runs} of \emph{S}, then $\psi_1 = \cup_{i=k_{2}+1}^{k_2+k_4} col(\rho_i)$, therefore  $ \psi_1 \subseteq$  $\{c_1, c_2, $ $\dots, c_{k_2}\}$ because  set of colours of all the \emph{first runs} is $\cup_{i=1}^{k_1} col(\rho_i)$ = $\{c_1, c_2, \dots, c_{k_2}\}$.. Without loss of generality, we assume that $\psi_{1}$= $\{c_1, c_2, $ $\dots, c_{k_5}\}$ $ \subseteq $ $\{c_1, c_2,$ $ \dots, c_{k_2}\}$ where $k_5 \leq k_2$, and $\psi_{2}=\cup_{i=k_5+1}^{k_2} col(\rho_i)$ = $\{c_{k_5+1}, c_{k_5+2},\dots, c_{k_2}\}$. We take \emph{first runs} from \emph{S} in the set $S_2=\{\rho_{k_5+1}, \rho_{k_5+2}, \dots, \rho_{k_2}\}$ whose colour set is $\psi_2$. For the first $k_5$ runs $\rho_1$, $\rho_2$, \dots, $\rho_{k_5}$, if the corresponding legs $L_1$, $L_2$,\dots, $L_{k_5}$ have a vertex $q$ such that col($q$) $\notin \psi$ then put the corresponding \emph{first run} in $S_3$. Mathematically we can say that for at least one vertex $q$ of the leg $L_i$, if col($q$) $\notin \psi$, then $\rho_i\in S_3$ for $i=1, 2, \dots, k_5$. Let $V^{'}=\{V_i$ : $\rho_i \in S_3$ and $leg(\rho_i) $= $L_i \}$.

Now we take $k_4$ vertices in $T_1=\{p_i : p_{k_2+i}\in \rho_{k_2+i} $ for $i=1,2,\dots, k_4 \}$ whose colour set is $\psi_1 =\{c_1, c_2, \dots, c_{k_5}\}$, $k_2-k_5$ vertices in $T_2=\{p_j$ $: p_j$ $\in \rho_{j} $ for $i=k_5+1, k_5+2,\dots, k_2 \}$ whose colour set is $\{c_{k_5+1}, c_{k_5+2}, \dots, c_{k_2}\}$, and $|S_3|$ vertices in $T_3=\{p_l : p_{l}\in \rho_{l} $  such that $   \rho_{l} \in S_3 $ for $i=1,2,\dots, |S_3| \}$ in such a way that hop-distances of the vertices of $T_1 \cup T_2 \cup T_3$ from $v$ are equal. We take each vertex of $T_1 \cup T_2 \cup T_3$ in $C$, then all the vertices of $V\backslash ((\cup_{i=k_2+1}^{k_2+k_4}V_i )\bigcup (\cup_{i=k_5+1}^{k_2}V_i)\bigcup V^{'}$) are covered with respect to their consistency. We set $T=T_1\cup T_2\cup T_3$ and $S_0=S_1\cup S_2\cup S_3$. We repeat the process in such a way that, for every $\alpha \in \{1,2,\dots, $ min $\{$length of the runs of $S_0\}\}$, hop-distances $(p_{\alpha},v)$ are equal. Now, $C_1$= $\cup_{\rho_i \in S_0}\hat C_i(p_{\alpha})$. We already defined $\hat C_i(p_{\alpha})$ in \ref{case2}. We set $\chi_{1}= |C_1|$. If $\chi > \chi_{1}$, then $C=C_1$ and $\chi = \chi_{1}$; otherwise, $C$ remains unchanged.

We repeat this process by taking $k_2$ runs from $k_1$ runs every time and update $C$ if $\chi > \chi_{1}$. We explain the whole procedure in algorithm \ref{alg:subcase03}. The time complexity will surely be $O($$\binom{k_1}{k_2})$ $*(n+ k*k_1 +k_1*n +n*n ))$. We know $k_2 < k$, $k \leq n$, and $k_1 < n$. Hence, $O($$\binom{k_1}{k_2}$ $*(n+ k*k_1 +k_1*n +n*n))$ $\approx O(n^{k+2})$.

After discussing all the cases and sub-cases, it is clear that the run time does not depend on the value of $k$ for the case \ref{case1} and case \ref{case2}; however, the run time for the case \ref{case3} is at most $O(n^{k+2})$ when $k$ is a constant. Hence we get the following theorem.

\begin{theorem}
The proposed algorithms correctly compute $C$ of a k-chromatic spider graph $G$ in $O(n^{k+2})$ time, where $n$ is the number of vertices in $G$ and $k (\ge 3)$ is a constant.
\end{theorem}

\section{Complexity result}

We reduce from the  MAX-2SAT problem.
We transform any large enough instance $\theta$ of MAX-2SAT into a tree $T$ of size polynomial in the size of $\theta$ such that $k$ clauses of $\theta$ can be satisfied if and only if $T$
has a consistent subset of a certain size that we describe later on.
The constructed tree $T$ is divided into variable and clause gadgets, which are then
connected via paths.
Throughout this section, by distance between two vertices, we refer to the graph theoretic or hop distance between them. The distance between a vertex $v$ and a path is the minimum distance between $v$ and any vertex of the path. The distance between two paths $P_1$ and $P_2$ is the minimum distance of any pair vertices $v_1$, $v_2$, where $v_1 \in P_1$, and $v_2 \in P_2$.
\subsection{Construction of the reduction tree}

Suppose that the given MAX-2SAT formula $\theta$ has $n$ variables $x_1, x_1, \ldots, x_n$ and $m$ clauses $c_1, c_2, \ldots, c_m$, for $n,m \geq 50$. We now describe the parts of its reduction tree $T$.
$\\ \\$
A variable gadget $X_i$ for $x_i$ has the following components (Figure \ref{fig:red}):
\begin{enumerate}
 \item A path $(x_{i,1}, x_{i,2}, \ldots, x_{i,5})$ on five vertices, called the  \emph{left literal path}, or the \emph{positive literal path} of the variable gadget $X_i$.
   A path $(\overline{x_{i,1}}, \overline{x_{i,2}}, \ldots, \overline{x_{i,5}})$ on five vertices, called   the \emph {right literal path}, or the \emph{positive literal path} of the variable gadget $X_i$.
  Each of the two literal paths has a distinct colour.

 \item Five paths of the form $(v_{i,j,1}, v_{i,j,2}, \ldots, v_{i,j,6})$, $1 \leq j \leq 6$
 on six vertices each, called the \emph{variable paths} of $x_i$.
Each of the five variable paths has a distinct colour.
 \item $mn$ pairs of vertices $\{s_{i,1,1}, s_{i,1,2}\}, \ldots, \{s_{i,mn,1}, s_{i,mn,2}\}$
 called the \emph{stabilizer vertices}. The vertices of the form $s_{i,j,1}$ are called the \emph{left stabilizers} while those of the form $s_{i,j,2}$ are called the right stabilizers.
 Each pair of stabilizer vertices has a distinct colour.
\end{enumerate}
These components are connected to each other and $T$ in the following manner.
\begin{enumerate}
 \item The first vertices of the variable paths, namely $v_{i,1,1}, v_{i,2,1}, v_{i,3,1}, v_{i,4,1}, v_{i,5,1}$ form an induced path in that order.
%
\item All the left stabilizer vertices are adjacent to the firt vertex of the first variable path ($v_{i,1,1}$). All the right stabilizer vertices are adjacent to the first vertex of the last variable path ($v_{i,5,1}$).
\item The first vertex of the left literal path is adjacent to the first vertex of the first
variable path. The first vertex of the right literal path is adjacent to the first vertex of the last variable path.
 \item The whole variable gadget $X_i$ is connected to the rest of $T$ via
the first vertex of the third variable path.
 \item
No colour of $X_i$ other than the colours $x_i$ and $\overline{x_i}$ occur
elsewhere in $T$. The occurrences of the colours $x_i$ and $\overline{x_i}$ are at least distance eight away from all vertices of $X_i$.
\end{enumerate}
A clause gadget $C_j$ of the clause $c_j$ has the following components (Figure \ref{fig:red}).
\begin{enumerate}
 \item A  path $(c_{j,1}, \ldots c_{j,17})$ on seventeen vertices, called the \emph{clause path}.
 The clause path has a distinct colour.
 \item Two paths $(o_{j,1,1}, o_{j,1,2}, \ldots, o_{j,1,15})$ and $(o_{j,2,1}, o_{j,2,2}, \ldots, o_{j,2,15})$ on fifteen vertices each, called the \emph{left occurrence path} and the \emph{right occurrence path} respectively.
 The two occurrence paths correspond to the two literals $l_1$ and $l_2$ that occur in $c_j$, and have the same colour as the corresponding literal paths of $l_1$ and $l_2$ in their respective variable gadgets.
\end{enumerate}
The first vertex of the left occurrence path is adjacent to the second vertex of the clause path. The first vertex of the right occurrence path is adjacent to the sixteenth vertex of the clause path. The clause gadget is joined to the rest of $T$ via the ninth vertex of the clause path. Overall, the variable and clause gadgets are joined as follows (Figure \ref{fig:red}).

\begin{enumerate}
 \item The first vertex of the third variable path of each variable gadget is joined to
 a vertex $p_0$ called the \emph{central chain vertex}, via an edge.
 \item The ninth vertex of the clause path of each clause gadget is joined to the central chain vertex via a path of three chain vertices.
 \item Each of the aforementioned chain vertices is made the first vertex of a path on eight
 vertices. These paths are called \emph{chain paths}.
\end{enumerate}
We denote by $\beta$ the total number of chain paths. Observe that $\beta = 3m + 1$.
We denote by $\alpha$ the total number of chain vertices. Observe that $\alpha = 8 \beta$

 \begin{figure}
\centering
\includegraphics[width=0.9\textwidth]{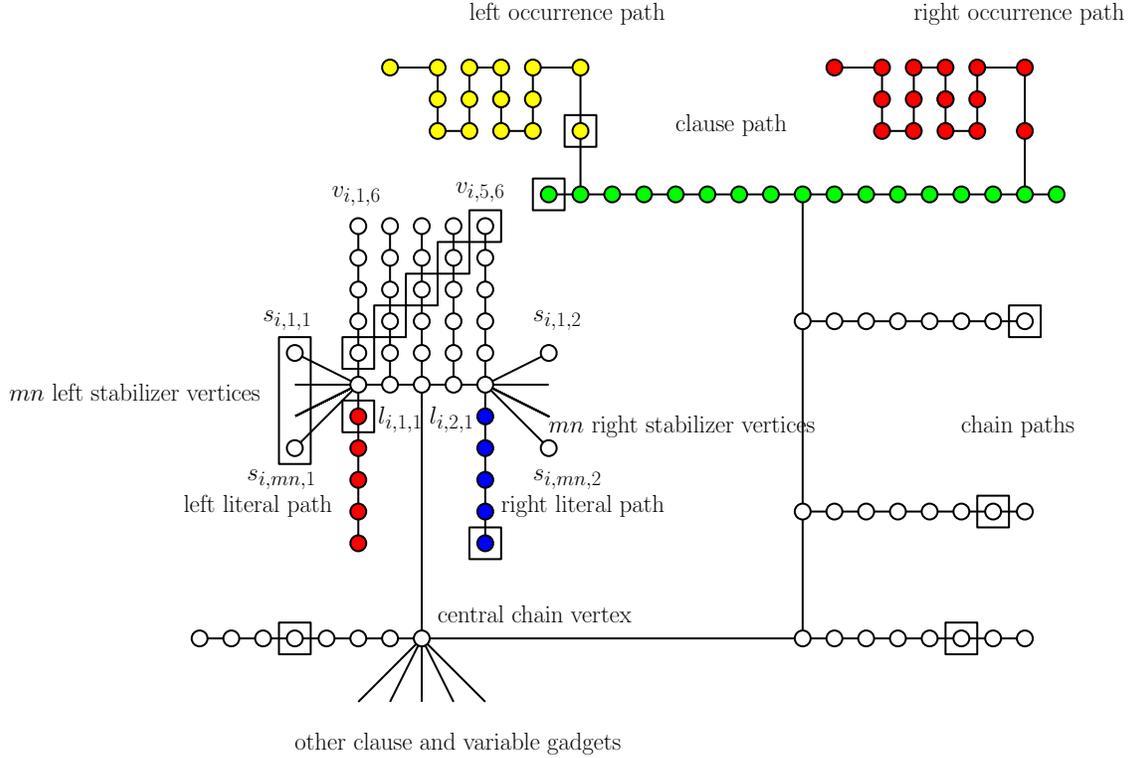}
\caption{\label{fig:red}  A part of the reduction tree. Only some paths are shown coloured, but the  stabilizers and other paths each has a distinct colour as well. The vertices chosen in an MCS are enclosed in squares. Since the left literal path has its first vertex in the MCS, and the right occurrence path has the same colour, the latter does not contribute to the MCS.}
\end{figure}
\subsection{Properties of the reduction tree}
In this section we discuss some properties of the reduction tree constructed in the previous section.

\begin{lemma} \label{lem:if}
 If $k$ clauses of $\theta$ can be satisfied, then $T$ has a consistent subset of size
$\beta + mn^2 + 7n  + 3m - k$.
\end{lemma}

\begin{proof}
Start with an empty subset $S$ of vertices.
 Consider an assignment that
satisfies $k$ clauses of $\theta$. If the variable $x_i$ is assigned $1$, then
for $S$ choose every left stabilizer vertex in the variable gadget of $x_i$.
From the $j^{th}$ variable path of the variable gadget of $x_i$, choose the $j+1^{th}$ vertex.
Choose the first vertex from the positive literal path, and the fifth vertex from the negative literal path.
$\\ \\$
On the other hand, if the variable $x_i$ is assigned $0$ then choose every right stabilizer vertex in the variable gadget of $x_i$. From the $j^{th}$ variable path of the variable gadget of $x_i$, choose the $7-j^{th}$ vertex. Choose the fifth vertex from the positive literal path, and the first vertex from the negative literal path. The variable gadgets together contribute $mn^2 + 7n$ vertices towards $S$.
$\\ \\$
Consider the clause $c_j$ with the literals $l_x$ and $l_y$. If none or both of $l_x$ and $l_y$ are assigned $1$, then choose any one leaf of the  clause path of $c_j$. If one of $l_x$ and $l_y$ is assigned $1$, then choose the leaf of the clause path closer to the occurrence path whose corresponding literal has been assigned $0$. In each case, choose the first vertex of the occurrence path from which the chosen leaf is at a distance $2$. Choose the last (fifteenth) vertex of the other occurrence path if and only if the corresponding literal has also been assigned $0$. The clause gadgets together contribute $3m - k$ vertices towards $S$.
$\\ \\$
Choose the fifth vertex of the chain path of $c_o$. For any other chain path, if its distance to the nearest clause path is $d$, then choose the $(d+8)^{th}$ vertex of the chain path. Thus the chain paths contribute $\beta$ vertices towards $S$. From the construction, it can be checked that $S$ is a consistent subset.

\end{proof}

\begin{lemma} \label{lem:path}
 If a monochromatic path is attached to the rest of $T$ with only one vertex,
 then it can contribute at most one vertex to a minimum consistent subset.
\end{lemma}
\begin{proof}
 Let $u$ be the first vertex of the path $p$, and let it be attached to $T \setminus p$
 via a vertex $v$. Suppose on the contrary to the claim, that two vertices $x$ and $y$ from $p$ are chosen in a minimum
 consistent subset $M$. Wlog let $x$ be closer to $u$ than $y$ is. Then removal of $y$ from $M$
 still gives a consistent subset, which contradicts the minimality of $M$.
\end{proof}

\begin{lemma} \label{lem:stbl}
A minimum consistent subset of $T$ has exactly one vertex from every pair of variable stabilizers. Furthermore, in a variable gadget, all the stabilizers that are in the minimum consistent subset are adjacent to a common vertex of a variable path.
\end{lemma}
\begin{proof}
 Since each pair of variable stabilizers has a unique colour, at least one of them must
 be in a consistent subset. Suppose both are in a consistent subset. Since they are only at a distance $2$ from other stabilizers on the same side, all stabilizers must be in the consistent subset as well. Similarly, if stabilizers of difference colours are chosen from
 both sides of the variable gadget, then all stabilizers of the variable gadget must be in the consistent subset. This means that the stabilizers contribute at least $(n+1)nm$, contradicting the minimality of the consistent subset, since due to the large values of $n$, $m$ and $\beta$, putting $k = m$ in Lemma \ref{lem:if}
 gives the maximum possible size of the minimum consistent subset, which still has fewer vertices.
\end{proof}

\begin{lemma} \label{lem:lpath}
 A minimum consistent subset of $T$ has exactly one vertex from every literal path.
 Moreover, in each variable gadget, any one of the first three vertices of one literal path,
 and the fifth vertex of the other literal path, must be in the minimum
 consistent subset.
\end{lemma}

\begin{proof}
Suppose that wlog, due to Lemma \ref{lem:stbl}, all the left stabilizer vertices of the variable gadget of $x_i$ are in a minimum consistent subset. These vertices are at a distance of $1$ and $6$ from the literal paths of the same variable gadget, respectively. Since the occurrence paths are at a distance of more than $6$ from the literal paths, each literal path
must have at least one vertex in the consistent subset. So due to Lemma \ref{lem:path}, each literal path has exactly one vertex in the minimum consistent subset.

Now consider the first vertex of the literal path of $x_i$. The left stabilizers, which are in the consistent subset, are at a distance $2$ from it.
This means that one of the first three vertices of the literal path
of $x_i$ must be in the minimum consistent subset. Now consider the right stabilizer vertices.
The nearest consistent subset vertex of the same colour is at a distance of $5$ from them. So if the literal path of $\overline{x_i}$ has any of its first four vertices in the consistent subset, then all the right stabilizers in the variable gadget also must be in the consistent subset, contradicting its minimality.
\end{proof}

\begin{lemma} \label{lem:only}
 If $T$ has a minimum consistent subset of size $\beta + mn^2 + 7n  + 3m - k$ then $\theta$ has an assignment satisfying
 at least $k$ clauses.
\end{lemma}
\begin{proof}
 Due to Lemma \ref{lem:path}, all the variable paths and chain paths contribute a vertex each to a minimum consistent subset $M$, totalling to $\beta + 5n$. The stabilizers contribute another $n^2m$, and the literal paths $2n$. By Lemma \ref{lem:lpath}, in each varible gadget, a literal path has any one of its first three vertices in $M$, and the other literal path has its fifth vertex in the consistent subset.
 Construct an assignment of $\theta$ as follows. If the literal path of $x_i$ has any one of its first three vertices in the consistent subset, then assign $1$ to $x_i$ and $0$ to $\overline{x_i}$.
 Otherwise, if the literal path of $x_i$ has its fifth vertex in the consistent subset, then assign $0$ to $x_i$ and $1$ to $\overline{x_i}$. Now we show that this assignment satisfies $k$ clauses of $\theta$. Since each clause line is of a distinct colour, they together contribute at least $m$ vertices to $M$. This means, the remaining at most $2m-k$ vertices of $M$ are contributed by occurrence paths. The choice of a consistent subset vertex in a clause path forces at least one of the occurrence paths to have a vertex in the consistent subset. This means that there are at least $k$ clause gadgets where exactly one occurrence line contributes to the consistent subset. The distance between any two occurrence paths from different clause gadgets is more than the distance between any occurrence path and the farther leaf of its clause path. So only the occurrence paths having the same colour as literal vertices whose third or previous vertices are in the consistent subset, may not contribute to the consistent subset. But these exactly correspond to the literals that have been assigned $1$, and hence each of the aforementioned $k$ clauses have at least one literal satisfied, proving the claim.
\end{proof}
The above lemmas immediately bring us to the following theorem.
\begin{theorem}
The Minimum Consistent Subset is NP-complete for trees.
\end{theorem}
\begin{proof}
 It is easy to see that the problem is in NP. As for  NP-hardness, Lemmas \ref{lem:if} and \ref{lem:only}
 establish a relationship between the number of satisfiable clauses of $\theta$ and the size of the minimum consistent subset of $T$, which can be obtained from $\theta$ in polynomial time.
\end{proof}
\section{Conclusion}

This paper presents a polynomial time algorithm for the minimum consistent subset of spider graphs where $k$ is constant. However, it remains unknown whether or not MCS can be computed for spider graphs in polynomial time irrespective of the number of colour classes. It is also shown that the problem is NP-complete for trees. However, the reduction uses polynomially many colour classes. The problem remains unsolved for $k$-chromatic trees for constant $k \ge 3$.

\bibliographystyle{unsrt}  
\bibliography{references}

\end{document}